\documentclass{amsart}
\usepackage{amsfonts}
\usepackage{amsthm}
\usepackage{euscript}
\usepackage{xcolor}
\usepackage{tikz}

\newtheorem{thm}{Theorem}
\newenvironment{thmbis}[1]
  {\renewcommand{\thethm}{\ref{#1}$'$}%
   \addtocounter{thm}{-1}%
   \begin{thm}}
  {\end{thm}}

  \newtheorem{prop}{Proposition}
\newenvironment{propbis}[1]
  {\renewcommand{\theprop}{\ref{#1}$'$}%
   \addtocounter{prop}{-1}%
   \begin{prop}}
  {\end{prop}}

\usepackage{float}

\usepackage{amsfonts,xcolor}
\usepackage{latexsym}
\usepackage{amsmath}
\usepackage{amssymb}
\usepackage{arydshln}

\def\rem{\mbox{rem}}
\usetikzlibrary{positioning,arrows,automata,shadows,petri}
\newtheorem{theorem}{Theorem}
\newtheorem{remark}{Remark}

\newtheorem{lemma}{Lemma}
\newtheorem{definition}{Definition}

\newtheorem{cor}{Corollary}
\newtheorem{proposition}{Proposition}
\newtheorem{example}{Example}

\usepackage{geometry}
\newcommand{\rom}[1]{%
  \textup{\uppercase\expandafter{\romannumeral#1}}%
}

\begin{document}

\title{Morphisms on infinite alphabets, countable states automata and regular sequences}
\author{Jie-Meng Zhang,Jin Chen, Yingjun Guo and Zhixiong Wen}




\begin{abstract}
In this paper, we prove that a class of regular sequences can be viewed as projections of fixed points of uniform morphisms on a countable alphabet, and also can be generated by countable states automata. Moreover, we prove that the regularity of some regular sequences is invariant under some codings.
\end{abstract}
\maketitle

\section{Introduction}
Morphisms on a finite alphabet are widely studied in many fields, such as finite automata, symbolic dynamics, formal languages, number theory and also in physics in relation to quasi-crystals (see \cite{jps,Jpa,M,VM,CKM}). A morphism is a map $\sigma:\Sigma^{*}\rightarrow\Sigma^{*}$ satisfying that $\sigma(uv)=\sigma(u)\sigma(v)$
for all words~$u,v\in\Sigma^{\ast}$, where $\Sigma^{*}$ is the free moniod generated by a finite alphabet $\Sigma$ (with $\epsilon$ as the neutral element). Naturally, the morphism can be extended to $\Sigma^{\mathbb{N}}$, which is the set of infinite sequences. The morphisms of constant length are called \emph{uniform morphisms} and the sequence $u=u(0)u(1)u(2)\cdots\in\Sigma^{\mathbb{N}}$ satisfying that $\sigma(u)=u$ is a \emph{fixed point} of $\sigma$.

In \cite{Cob}, Cobham showed that a sequence is a fixed point of a uniform morphism (under a coding) if and only if it is an automatic sequence. Recall that we call a sequence is \emph{automatic} if it can be generated by a finite state automaton. Moreover, a sequence $\{u(n)\}_{n\geq0}$ is $k$-automatic if and only if its $k$-kernel is finite, where the \emph{$k$-kernel} is defined by the set of subsequences,
$$\left\{\{u(k^in+j)\}_{n\geq0}:i\geq0,0\leq j\leq k^i-1\right\}.$$

However, the range of automatic sequences is necessarily finite. To overcome this limit, 
 Allouche and Shallit \cite{Jp} introduced a more general class of regular sequences which take their values in a (possibly infinite) Noetherian ring $R$.
Formally, a sequence is \emph{$k$-regular} if the module generated by its $k$-kernel is finitely generated.

Many regular sequences have been found and studied in \cite{jp,JKR,Bell,Sum,JH,AFO}. 
If a sequence $\{u(n)\}_{n\geq0}$ takes finitely many values, Allouche and Shallit showed that it is regular if and only if it is automatic in \cite{Jp}. Hence, we always assume that regular sequences take infinitely many values.
If a sequence $\{u(n)\}_{n\geq0}$ is an unbounded integer regular sequence,  Allouche and Shallit  \cite{Jp} proved that there exists $c_1\geq0$ such that $u(n)=O(n^{c_1})$ for all $n$ and Bell et. al. \cite{Minimal} showed that there exists $c_2\geq0$ such that $|u(n)|>c_2\log n$ infinitely often.
Recently, Charlier et. al. characterized the regular sequences by counting the paths in the corresponding graph with finite vertices in \cite{ENJ}.

Despite all this, there are no descriptions for regular sequences by automata. Note that automatic sequences can be generated by finite state automata, it is a natural question that can regular sequences be generated by automata with countable states, or morphisms on a countable alphabet? 

Morphisms on infinite alphabets and countable states automata have been studied by many authors.
In \cite{M06}, Mauduit concerned the arithmetical and statistical properties of sequences generated by deterministic countable states automata or morphisms on a countable alphabet. Meanwhile, Ferenczi \cite{F06} studied morphism dynamical systems on infinite alphabets and Le Gonidec \cite{LG,LG06,LG081,LG082} studied complexity function for some $q^{\infty}$-automatic sequences. More about morphisms on infinite alphabets and countable states automata, please see in \cite{T02,GS09,CG10}



In the present paper, we focus on morphisms on a countable alphabet and automata with countable states. We find a class of automata with countable states which can generate regular sequences. That is to say, a class of regular sequences can be generated by countable states automata.


This paper is organized as follows. In Section 2, we give some notations and definitions.
In Section 3, we introduce a class of morphisms on infinite alphabets and countable states automata, which are called to be $m$-periodic~$k$-uniform morphism and $m$-periodic~$k$-DCAO, respectively.  We prove that all the infinite sequences generated by them are $k$-regular.
In Section 4, we consider the codings generated by the sequences satisfying a linear recurrence. Under some conditions, we show that the regularity is invariant under these codings. In the last section, we outline some generalizations.

\section{Preliminary}
Let $\mathbb{N}^{\geq2}$ be the set of integers greater than $2$. For every integer~$b\in \mathbb{N}^{\geq2}$, we define a nonempty alphabet~$\Sigma_b :=\{0,1,\cdots,b-1\}$ and a countable alphabet $\Sigma_{\infty}:=\mathbb{N}=\{0,1,\cdots,n,\cdots\}$. For $b\in \mathbb{N}^{\geq2}\cup\{\infty\}$, let $\Sigma_b^{\ast}$ denote the set of all finite words on $\Sigma_b$. If $w\in\Sigma_b^{\ast}$, then its length is denoted  by $|w|$. If $|w|=0$, then we call $w$ is the empty word, denoted by $\epsilon$. Let $\Sigma_b^k$ denote the set of words of length~$k$ on $\Sigma_b$, i.e., $\Sigma_b^{\ast}=\bigcup_{k\geq0}\Sigma_b^k$. Let $u=u(0)u(1)\cdots u(m)$ and $v=v(0)v(1)\cdots v(n)\in\Sigma_b^*$. The word $uv:=u(0)u(1)\cdots u(m)v(0)v(1)\cdots v(n)$ denotes their \emph{concatenation}. If $|u|\geq1$~(resp.~$|v|\geq1)$, then $u$~(resp.~$v$) is a \emph{prefix} (resp. \emph{suffix}) of $uv$. Clearly, the set $\Sigma_b^{\ast}$ together with the concatenation forms a monoid, where the empty word $\epsilon$ plays the role of the neutral element.

If $b\in \mathbb{N}^{\geq2}$, then every non-negative integer $n$~has a unique representation of the form $n=\sum_{i=0}^{l}n_i b^i$ with
$n_l\neq0$~and~$n_i\in\Sigma_b$. We call $n_ln_{l-1}\cdots n_0$~its~\emph{canonical representation in base $b$}, denoted by~$(n)_b$. If $l\geq |(n)_b|$, denote $(n)_b^{l}=0^{i}(n)_b$ with $i=l-|(n)_b|$. If~$(n)_b=n_ln_{l-1}\cdots n_0$, then the \emph{base-$b$ sum of digits function} is defined by $s_b(n):=\sum_{i=0}^ln_i$. If $b\in \mathbb{N}^{\geq2}$ and $\mathbf{w}=w_lw_{l-1}\cdots w_0$, then $[\mathbf{w}]_b:=\sum_{i=0}^{l}w_i\cdot b^{i}.$
We denote by $\rem_b(n):=r$ if $n\equiv r~(\bmod~ b)~(0\leq r\leq b-1)$.

In this paper, unless otherwise stated, all alphabets under consideration are countable.
\subsection{Morphisms on countable alphabets}
Let~$\Sigma$~and~$\Delta$~be two alphabets. A \emph{morphism~(or substitution)} is a map~$\sigma$~from
~$\Sigma^{\ast}$~to~$\Delta^{\ast}$ satisfying that
$\sigma(uv)=\sigma(u)\sigma(v)$
for all words~$u,v\in\Sigma^{\ast}$. In the whole paper, we use the term ``morphism".

Note that $\sigma(\epsilon)=\epsilon$. If $\Sigma=\Delta$, then we can iterate the application of $\sigma.$ That is, $\sigma^{i}(a)=\sigma(\sigma^{i-1}(a))$ for all $i\geq1$ and $\sigma^{0}(a)=a.$

Let $\sigma$ be a morphism defined on $\Sigma=\{q_0,q_1,\cdots,q_n,\cdots\}$. If $\sigma(q_i)=q_{i_1}q_{i_2}\cdots q_{i_{t_i}}$ with $i_j=a_j i+b_j$ and $a_j,b_j\in\mathbb{Z},$ 
for every $i\geq0$, then $\sigma$ is called a \emph{linear morphism}. If there exists some integer $k\geq1$ such that $|\sigma(a)|=k$ for all $a\in\Sigma$,  then $\sigma$ is called a \emph{$k$-uniform morphism (or $k$-constant length morphism)}. A~$1$-uniform morphism is called a \emph{coding.} If there exists a finite or infinite word $w\in\Sigma^{\mathbb{N}}$ such that $\sigma(w)=w$, then the word $w$ is a \emph{fixed point} of $\sigma$. In fact, if $\sigma(a)=aw$ for some letter $a\in\Sigma$ and nonempty $w\in\Sigma^{*}$, then the sequence of words $a,\sigma(a),\sigma^{2}(a),\cdots$ converges to the infinite word $$\sigma^{\infty}(a):=aw\sigma(w)\sigma^{2}(w)\cdots,$$
where the limit is defined by the metric $d(u,v)=2^{-min\{i:u(i)\neq v(i)\}}$ for $u=u(0)u(1)\cdots$ and
$v=v(0)v(1)\cdots$. Clearly, $\sigma^{\infty}(a)$ is a fixed point of $\sigma.$ Hence, for every morphism $\sigma$ on the alphabet $\Sigma$, we always assume that there exists a letter $a\in\Sigma$ such that $\sigma(a)=aw$ with a nonempty word $w\in\Sigma^{*}$.

\begin{example}\label{ex1}
Let $\Sigma=\Sigma_{\infty}:=\{0,1,\cdots,n,\cdots\}$. Define a 2-uniform morphism $\sigma(i)=i(i+1)$ for all $i\geq0$, then $\sigma^{\infty}(i)=i(i+1)(i+1)(i+2)\cdots$ is a fixed point of $\sigma$. In particular, the fixed point $\sigma^{\infty}(0)=01121223\cdots$ is the sequence of base-$2$ sum of digits function $\{s_2(n)\}_{n\geq0}$.
\end{example}

\begin{example}\label{ex2}(The drunken man morphism)
Let $\Sigma=\{\iota\}\cup\mathbb{Z}$.  Define a 2-uniform morphism $\sigma(\iota)=\iota 1$ and $\sigma(i)=(i-1)(i+1)$ for all $i\in\mathbb{Z}$, then the infinite word $\sigma^{\infty}(\iota)=\iota 102(-1)113(-2)0020224\cdots$ is the only non-empty fixed point of $\sigma$.
\end{example}

\begin{example}\label{ex3}(Infinibonacci morphism)
Let $\Sigma=\mathbb{N}$. Define a 2-uniform morphism $\sigma(i)=0(i+1)$ for all $i\geq0$, then  $\sigma^{\infty}(0)=0102010301020104\cdots$ is a fixed point of $\sigma$.
\end{example}

\subsection{Deterministic infinite states automata}
A \emph{deterministic countable automaton (DCA)} is a directed graph $M=(Q,\Sigma,\delta,q_0,F),$ where $Q$~is a countable set of states, $q_0\in Q$~is the initial state, $\Sigma$~is the finite input alphabet, $F\subseteq Q$~is the set of accepting states,  $\delta:Q\times\Sigma\rightarrow Q$~is the transition function. The transition function $\delta$ can be extended to  $Q\times\Sigma^{*}$ by $\delta(q,\epsilon)=q$ and $\delta(q,wa)=\delta(\delta(q,w),a)$ for all $q\in Q,a\in\Sigma$ and $w\in\Sigma^{*}$.

Similarly, a \emph{deterministic countable states automaton with output (DCAO)} is defined to be a~$6$-tuple $M=(Q,\Sigma,\delta,q_0,\Delta,\tau),$ where~$Q,\Sigma,\delta,q_0$~are as in the definition of DCA as above, $\Delta$~is the output alphabet and $\tau:~Q\rightarrow\Delta$~is the output function.
In particular, if the input alphabet $\Sigma=\Sigma_k$ for some $k\in\mathbb{N}^{\geq2}$, then the automaton $M$ is always called to be a $k$-DCAO.

Let $\{u(n)\}_{n\geq0}=u(0)u(1)u(2)\cdots$ be a sequence on the alphabet $\Delta$. The sequence $\{u(n)\}_{n\geq0}$ is called to be \emph{$k$-automatic}, if the sequence can be generated by a $k$-DCAO, that is, there exists a $k$-DCAO $M=(Q,\Sigma_k,\delta,q_0,\Delta,\tau)$ such that $u(n)=\tau(\delta(q_0,w))$ for all $n\geq0,w\in\Sigma_{k}^{*}$ and $[w]_k=n$.

By the choice of DCAO $M$ satisfying that $\delta(q_0,0)=q_0$, our machine $M$ always computes the same $u(n)$  even if the input one has leading zeros. Hence, unless otherwise stated, all DCAOs satisfy $\delta(q_0,0)=q_0$ and $u(n)=\tau(\delta(q_0,(n)_k))$ for all $n\geq0$.
\begin{example}\label{ex4}
Let $Q=\{q_0,q_1,q_2,\cdots\},\Delta=\mathbb{N}, \delta(q_i,0)=q_i,\delta(q_i,1)=q_{i+1}$ and $\tau(q_i)=i$ for all $i\geq0$. Then, the sequence of base-$2$ sum of digits function $\{s_2(n)\}_{n\geq0}$ is $2$-automatic. It can be generated by a $2$-DCAO in Figure \ref{aut1}.
\begin{figure}[H]
\centering
\begin{tikzpicture}[scale=0.8, every node/.style={scale=0.8}, state/.style={scale=0.8, circle solidus,draw,
inner sep=1pt,minimum size=12mm},node distance = 2.5cm,>=stealth,->,auto,black]
    \node[state,initial]  (q_0)                      {$q_0$ \nodepart{lower} $0$};
    \node[state] (q_1) [right of=q_0] {$q_1$ \nodepart{lower} $1$};
    \node[state] (q_2) [right of=q_1] {$q_2$ \nodepart{lower} $2$};
    \node(q_3) [right of=q_2] {$\cdots$};
    \node[state] (q_n) [right of=q_3] {$q_n$ \nodepart{lower} $n$};
    \node(q_{n+1}) [right of=q_n] {$\cdots$};

    \path[->]
    		(q_0) edge [bend left]  node {$1$} (q_1)
                      edge [loop above] node {$0$} ()
    		(q_1) edge [bend left]   node {$1$} (q_2)
                      edge [loop above] node {$0$} ()
    		(q_2) edge [bend left]   node  {$1$} (q_3)	
    		          edge [loop above] node {$0$} ()
            (q_3) edge [bend left]   node  {$1$} (q_n)	
            (q_n) edge [bend left]   node  {$1$} (q_{n+1})
                      edge [loop above] node {$0$} ();
    \end{tikzpicture}
\caption{DCAO generating the base-$2$ sum of digits function.}
\label{aut1}
\end{figure}
\end{example}

\begin{example}\label{ex5}
Let $Q=\{q_0\}\cup\mathbb{Z},\Delta=\{\iota\}\cup\mathbb{Z}, \delta(q_0,0)=q_0,\delta(q_0,1)=1,\delta(i,0)=i-1,\delta(i,1)=i+1$, $\tau(q_0)=\iota$ and $\tau(i)=i$ for all $i\in\mathbb{Z}$. Then, the sequence defined in Example \ref{ex2} can be generated by a $2$-DCAO in Figure \ref{aut2}.
\begin{figure}[H]
\centering
\begin{tikzpicture}[scale=0.8, every node/.style={scale=0.8}, state/.style={scale=0.8, circle solidus,draw,
inner sep=1pt,minimum size=12mm},>=stealth,->,auto,black]
\node[state] (a) {$1$ \nodepart{lower} $1$};
\node[state] (b) [right=15mm of a] {$2$ \nodepart{lower} $2$};
\node(c) [right=15mm of b] {$\cdots$};
\node[state] (e) [left=15mm of a] {$0$ \nodepart{lower} $0$};
\node[state] (f) [left=15mm of e] {$-1$ \nodepart{lower} $-1$};
\node(g) [left=15mm of f] {$\cdots$};
\node[state,initial] (i)[above=8mm of a]  {$q_0$ \nodepart{lower} $\iota$};
\draw [->] (c.200) to [bend left] node [below] {$0$} (b.-20);
\draw [->] (b.20) to [bend left]  node [above] {$1$} (c.160);
\draw [->] (b.200) to [bend left] node [below] {$0$} (a.-20);
\draw [->] (a.20) to [bend left]  node [above] {$1$} (b.160);
\draw [->] (a.200) to [bend left] node [below] {$0$} (e.-20);
\draw [->] (e.20) to [bend left]  node [above] {$1$} (a.160);
\draw [->] (e.200) to [bend left] node [below] {$0$} (f.-20);
\draw [->] (f.20) to [bend left]  node [above] {$1$} (e.160);
\draw [->] (f.200) to [bend left] node [below] {$0$} (g.-20);
\draw [->] (g.20) to [bend left]  node [above] {$1$} (f.160);
\draw [->] (i.270) to node [right] {$1$} (a.90);
\draw [->] (i) to [out=120, in=60,loop,distance=10mm] node [above] {$0$} (i);
\end{tikzpicture}

\caption{DCAO generating the sequence defined in Example \ref{ex2}.}
\label{aut2}
\end{figure}
\end{example}

\begin{example}\label{ex6}
Let $Q=\{q_0,q_1,q_2,\cdots\},\Delta=\mathbb{N}, \delta(q_i,0)=q_0,\delta(q_i,1)=q_{i+1}$ and $\tau(q_i)=i$ for all $i\geq0$. Then, the sequence defined in Example \ref{ex3} is $2$-automatic. It can be generated by a $2$-DCAO in Figure \ref{aut3}.
\begin{figure}[H]
\centering
\begin{tikzpicture}[scale=0.8, every node/.style={scale=0.8}, state/.style={scale=0.8, circle solidus,draw,
inner sep=1pt,minimum size=12mm},>=stealth,->,auto,black]
\node[state,initial] (a) {$q_0$ \nodepart{lower} $0$};
\node[state] (b) [right=15mm of a] {$q_1$ \nodepart{lower} $1$};
\node[state] (c) [right=15mm of b] {$q_2$ \nodepart{lower} $2$};
\node[state] (d) [right=15mm of c] {$q_3$ \nodepart{lower} $3$};
\node (e) [right=15mm of d]{$\cdots$};

\draw [->] (a) to [out=120, in=60,loop,distance=10mm] node [above] {$0$} (a);
\draw [->]  (a.20) to [bend left] node [above] {$1$} (b.160);
\draw [->]  (b.20) to [bend left] node [above] {$1$} (c.160);
\draw [->]  (c.20) to [bend left] node [above] {$1$} (d.160);
\draw [->]  (d.20) to [bend left] node [above] {$1$} (e.160);
\draw [->]  (b.200) to [bend left] node [above] {$0$} (a.-20);
\draw [->]  (c.220) to [bend left=40] node [above] {$0$} (a.-40);
\draw [->]  (d.240) to [bend left=43] node [above] {$0$} (a.-60);
\end{tikzpicture}

\caption{DCAO generating the sequence defined in Example \ref{ex3}.}
\label{aut3}
\end{figure}
\end{example}

\bigskip

By the definitions of $k$-uniform morphism and $k$-DCAO,  we note that each sequence $\mathbf{u}=\{u(n)\}_{n\geq0}$ can be generated by a $k$-uniform morphism or a $k$-DCAO for every $k\in\mathbb{N}^{\geq2}$. In fact, let $\sigma:\mathbb{N}\rightarrow \mathbb{N}^{*}$ be a $k$-uniform morphism  defined by $\sigma(i)=(ki)(ki+1)\cdots (ki+k-1)$ and $\rho$ be a coding by $\rho(i)=u(i)$. We have $\mathbf{u}=\rho(\sigma^{\infty}(0))$. Similarly, it can be generated by a $k$-DCAO by choosing $\delta(q_0,(n)_k)=q_n$ and $\tau(q_n)=u(n)$ for all $n\geq0$.

Hence, in the whole paper, we will consider the $k$-uniform morphisms and the $k$-DCAOs that the codings and the output functions satisfying some conditions. Assume that $\sigma$ is a morphism on the alphabet $\{q_0,q_1,q_2,\cdots\}$,  we will consider the $k$-DCAO satisfying $\tau(q_i)=i$ in Section 3,  $\tau(q_i)=L_i$ in Section 4  and some others, where $\{L_i\}_{i\geq0}$ is a given sequence.

\section{Regularity of the index sequences generated by~$m$-periodic~$k$-uniform morphisms}
In this section, we first introduce a definition in the following.
\begin{definition}\label{periodicmorphism}
If there exists a matrix $T=(t_{r,s})_{0\leq r\leq m-1,0\leq s\leq k-1}\in \mathbb{N}^{m\times k}$ such that
\begin{equation*}\label{periodic}
  \sigma(q_{mi+n})=q_{mi+t_{n,0}}q_{mi+t_{n,1}}~\cdots~q_{mi+t_{n,k-1}}
\end{equation*}
for every $0\leq n\leq m-1$ and $i\geq0$,
then $\sigma$ is called to be a \emph{$m$-periodic~$k$-uniform morphism} on $\{q_0,q_1,q_2,\cdots\}$ and we always denote it by $\sigma_{T}$. The matrix $T$ is called to be the \emph{index matrix} of $\sigma.$
\end{definition}

If there exists an integer $0\leq n\leq m-1$ such that $t_{n,0}=n$, then $\sigma_T^{\infty}(q_{n})$ is a fixed point of $\sigma_T$. Unless otherwise stated, we assume~$t_{0,0}=0$, then $\sigma_T^{\infty}(q_0)$ is a fixed point of $\sigma_T$.
By the definition of $m$-periodic~$k$-uniform morphism, we have the following definition similarly.
\begin{definition}\label{periodicIDCA}
If there exists a matrix $T=(t_{r,s})_{0\leq r\leq m-1,0\leq s\leq k-1}\in \mathbb{N}^{m\times k}$ such that the $k$-DCAO $M=(Q,\Sigma_k,q_0,\delta,\Delta,\tau)$ satisfying that
$$\delta(q_{mi+n},j)=q_{mi+t_{n,j}},$$
for every $0\leq n\leq m-1,0\leq j\leq k-1$ and $i\geq0$,
then the automaton $M$ is called to be a \emph{$m$-periodic~$k$-DCAO} and we always denote it by $M_{T}$. 
\end{definition}

Clearly, the fixed point $\tau(\sigma_T^{\infty}(q_0))$ can be generated by the $k$-DCAO $M_T$.
In the rest of this section, we will consider the coding $\tau:q_i\rightarrow i~(i\geq0)$ and its corresponding sequence $\{\mathbf{i}(n)\}_{n\geq0}:=\tau(\sigma_T^{\infty}(q_0))$. Note that $$\sigma_T^{\infty}(q_0)=\{q_{\mathbf{i}(n)}\}_{n\geq0}.$$
Hence, the sequence~$\{\mathbf{i}(n)\}_{n\geq0}$ is called to be the \emph{index sequence} of the morphism $\sigma_T$.

In fact, the index sequence~$\{\mathbf{i}(n)\}_{n\geq0}$ can be generated by the morphism $\sigma^{\prime}$ defined on $\mathbb{N}$ by $$\sigma^{\prime}(mi+n)=(mi+t_{n,0})~(mi+t_{n,1})~\cdots~(mi+t_{n,k-1}),$$ for every $0\leq n\leq m-1$ and $i\geq0$.
Similarly, the corresponding transition function can be defined by $\delta(r,s)=t_{r,s}~(0\leq r\leq m-1,0\leq s\leq k-1)$
and $\delta(mi+j,a)=mi+\delta(j,a)$, for every $i\geq0$, $j\in\Sigma_m$ and $a\in\Sigma_k$.
It is easy to check that for every $n\geq0$, $\mathbf{i}(n)=\delta(0,(n)_{k})$ and $\delta(n,w)=n+\delta(\rem_m(n),w)-\rem_m(n)$ for every $w\in\Sigma_k^{\ast}$.

\begin{lemma}\label{transition function}
For every $l\geq0$ and $0\leq j\leq k^l-1$, we have
$$\mathbf{i}(k^ln+j)=\mathbf{i}(n)+\delta(\rem_m(\mathbf{i}(n)),(j)_{k}^{l})-\rem_m(\mathbf{i}(n)).$$
\end{lemma}

\begin{proof}
For every $l\geq0$ and $0\leq j\leq k^l-1$, we denote $(j)_{k}^{l}:=j_{l-1}j_{l-2}\cdots j_1j_0$. Then,
\begin{eqnarray*}
   \mathbf{i}(k^ln+j)  &=& \delta(0,(k^ln+j)_k)=\delta(0,(n)_kj_{l-1}j_{l-2}\cdots j_1j_0) \\
 &=& \delta(\delta(0,(n)_k),j_{l-1}j_{l-2}\cdots j_1j_0)=\delta(\mathbf{i}(n),j_{l-1}j_{l-2}\cdots j_1j_0)  \\
     &=& \mathbf{i}(n)+\delta(t,j_{l-1}j_{l-2}\cdots j_1j_0)-t
\end{eqnarray*}
where $t=\rem_m(\mathbf{i}(n)).$
\end{proof}

By Lemma \ref{transition function}, we will prove the following theorem.
\begin{theorem}\label{index regular}
 If~$\sigma$~is a~$m$-periodic~$k$-uniform morphism, then the index sequence~$\{\mathbf{i}(n)\}_{n\geq0}$~is~$k$-regular.
\end{theorem}
\begin{proof}
For every $l\geq0$ and $0\leq j\leq k^l-1$, let $(j)_{k}^{l}=j_{l-1}j_{l-2}\cdots j_1j_0$. Define
$$V_{l,j}:=\left(
    \begin{array}{l}
      \delta(0,j_{l-1}j_{l-2}\cdots j_1j_0) \\
      \delta(1,j_{l-1}j_{l-2}\cdots j_1j_0)-1\\
      \delta(2,j_{l-1}j_{l-2}\cdots j_1j_0)-2\\
     \cdots\cdots\cdots\\
     \delta(m-1,j_{l-1}j_{l-2}\cdots j_1j_0)-(m-1)\\
    \end{array}
  \right).
$$

By Lemma \ref{transition function},
if $V_{l,j}=c_1V_{l^{\prime},j^{\prime}}+c_2V_{l^{\prime\prime},j^{\prime\prime}}$ $(l^{\prime},l^{\prime\prime}<l)$, then $\mathbf{i}(k^ln+j)=c_1\mathbf{i}(k^{l^{\prime}}n+j^{\prime})+c_2\mathbf{i}(k^{l^{\prime\prime}}n+j^{\prime\prime})-(c_1+c_2-1)\mathbf{i}(n)$. Note from \cite{jp} that if a sequence satisfies linear recurrence relations, then it is regular. Hence, it suffices to prove that the vectors satisfy the linear recurrence relations.

By the definition of the vector~$V_{l,j}$, we have
\begin{eqnarray*}
 V_{l,j} &=&  \left(
                  \begin{array}{c}
                   \delta(\delta(0,j_{l-1}j_{l-2}\cdots j_1),j_0) \\
                     \delta(\delta(1,j_{l-1}j_{l-2}\cdots j_1),j_0)-1 \\
 \delta(\delta(2,j_{l-1}j_{l-2}\cdots j_1),j_0)-2\\
\vdots\\
 \delta(\delta(m-1,j_{l-1}j_{l-2}\cdots j_1),j_0)-(m-1)\\
                  \end{array}
                \right)\\
   &=&  \left(
                  \begin{array}{c}
                   \delta(0,j_{l-1}j_{l-2}\cdots j_1)+\delta(l_0,j_0)-l_0 \\
                     \delta(1,j_{l-1}j_{l-2}\cdots j_1)-1+\delta(l_1,j_0)-l_1 \\
 \delta(2,j_{l-1}j_{l-2}\cdots j_1)-2+\delta(l_2,j_0)-l_2\\
\vdots\\
 \delta(m-1,j_{l-1}j_{l-2}\cdots j_1)-(m-1)+\delta(l_{m-1},j_0)-l_{m-1}\\
                  \end{array}
                \right)\\
   &=& V_{l-1,\frac{j-j_0}{k}}+\left(
                  \begin{array}{c}
                   \delta(l_0,j_0)-l_0 \\
                     \delta(l_1,j_0)-l_1 \\
 \delta(l_2,j_0)-l_2\\
\vdots\\
\delta(l_{m-1},j_0)-l_{m-1}\\
                  \end{array}
                \right),
\end{eqnarray*}
where~$l_i=\rem_m(\delta(i,j_{l-1}j_{l-2}\cdots j_1))$~for~$0\leq i\leq m-1$. Since $l_0,l_1,\cdots,l_{m-1}\in\Sigma_m$ and $j_0\in\Sigma_k$, there are at most~$m^{mk}$~different vectors
on the last equation. Hence, there exists an integer~$L$ such that for all~$l>L$ and $0\leq j\leq k^{l}-1$, there exist integers $l^{\prime}\leq L$ and $0\leq j^{\prime}\leq k^{l^{\prime}}-1$ satisfying that
$$V_{l,j}-V_{l-1,\frac{j-j_0}{k}}=V_{l^{\prime},j^{\prime}}
-V_{l^{\prime}-1,\frac{j^{\prime}-j_0^{\prime}}{k}}$$
where~$j_0,j_0^{\prime}$~are the least digits of the~$k$-ary
expansion of~$j$ and $j^{\prime}$ respectively.
Hence, each vector~$V_{l,j}$~is
a linear combination  of the vectors $V_{l^{\prime},j^{\prime}}$ with $l^{\prime}\leq L$ and $0\leq j^{\prime}\leq k^{l^{\prime}}-1$.
It implies that each sequence $\{\mathbf{i}(k^ln+j)\}_{n\geq 0}$ is
a linear combination of the sequences $\{\mathbf{i}(k^{l^{\prime}}n+j^{\prime})\}_{n\geq 0}$
with $l^{\prime}\leq L$ and $0\leq j^{\prime}\leq k^{l^{\prime}}-1$. Thus, the index sequence $\{\mathbf{i}(n)\}_{n\geq0}$~is~$k$-regular.
\end{proof}

By Theorem 2.5 in \cite{Jp}, if the  integer sequences $\{u(n)\}_{n\geq0}$ and $\{v(n)\}_{n\geq0}$ are both $k$-regular sequences, then $\{u(n)+v(n)\}_{n\geq0},\{au(n)\}_{n\geq0}$ and $\{u(n)v(n)\}_{n\geq0}$ are also $k$-regular. Hence, we have the following proposition.
\begin{proposition}\label{main2}
Take $\tau:q_i\rightarrow f(i)$~for every $i\geq0$, where
$f(i)$ is a polynomial with integer coefficients. If~$\sigma$~is a~$m$-periodic~$k$-uniform morphism, then the sequence~$\tau(\sigma^{\infty}(q_0))$
is $k$-regular.
\end{proposition}

\begin{proof}
Since $\sigma^{\infty}(q_0)=\{q_{\mathbf{i}(n)}\}_{n\geq0}$, we have
$$ \tau(\sigma^{\infty}(q_0))=\tau(\{q_{\mathbf{i}(n)}\}_{n\geq0})=\{f(\mathbf{i}(n))\}_{n\geq0}.$$
Hence, by Theorem \ref{index regular} and Theorem 2.5 in \cite{Jp}, both $\{\mathbf{i}(n)\}_{n\geq0}$ and $\{f(\mathbf{i}(n))\}_{n\geq0}$~are $k$-regular.
\end{proof}

The following example shows that the periodic condition in Theorem \ref{index regular} is necessary.
\begin{example}\label{ex7}
The index sequence of $\sigma_1:q_0\rightarrow q_0q_1$, $q_i\rightarrow q_{i-1}q_{i+1}~(i\geq1)$ is not $2$-regular. It can be generated by a DCAO in Figure \ref{AUT4}.
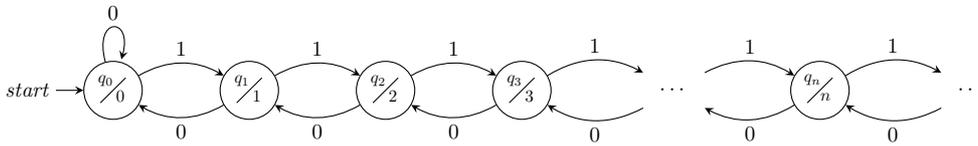
\begin{figure}[H]
\centering
\begin{tikzpicture}[scale=0.8, every node/.style={scale=0.8}, state/.style={scale=0.8, circle solidus,draw,
inner sep=1pt,minimum size=12mm},>=stealth,node distance=2.8cm,->,auto,black]
    \node[state,initial]  (q_0)                      {$q_0$ \nodepart{lower} $0$};
    \node[state] (q_1) [right of=q_0] {$q_1$ \nodepart{lower} $1$};
    \node[state] (q_2) [right of=q_1] {$q_2$ \nodepart{lower} $2$};
    \node[state] (q_3) [right of=q_2] {$q_3$ \nodepart{lower} $3$};
    \node(q_4) [inner sep=1pt,minimum size=10mm,right of=q_3,xshift=-0.3cm] {$\cdots$};
    \node[state] (q_n) [right of=q_4,xshift=0.2cm] {$q_n$ \nodepart{lower} $n$};
    \node(q_{n+1}) [inner sep=1pt,minimum size=10mm,,right of=q_n,xshift=-0.3cm] {$\cdots$};
    \path[->]
    		(q_0) edge [bend left]  node {$1$} (q_1)
                      edge [loop above] node {$0$} ()
    		(q_1) edge [bend left]   node {$0$} (q_0)
                      edge [bend left]  node {$1$} (q_2)
    		(q_2) edge [bend left]   node  {$1$} (q_3)	
    		          edge [bend left] node {$0$} (q_1)
            (q_3) edge [bend left]   node  {$1$} (q_4)	
    		          edge [bend left] node {$0$} (q_2)
            (q_4) edge [bend left]  node  {$1$} (q_n)
                      edge [bend left] node {$0$} (q_3)
            (q_n) edge [bend left] node {$1$} (q_{n+1})
                      edge [bend left] node {$0$} (q_4)
            (q_{n+1}) edge [bend left] node {$0$} (q_n);
    \end{tikzpicture}
\caption{A DCAO generating the index sequence of $\sigma_1$.}
\label{AUT4}
\end{figure}
\end{example}

Assume that the index sequence $\{\mathbf{i}(n)\}_{n\geq0}$ of $\sigma_1$~is a $2$-regular sequence, then $\{\rem_2(\mathbf{i}(n))\}_{n\geq0}$ is $2$-automatic, which implies that the set $\left\{\{\rem_2(\mathbf{i}(2^kn))\}_{n\geq0}: k\geq0\right\}$
is finite. Hence, there exist two integers $0<k_1<k_2$ such that $\rem_2(\mathbf{i}(2^{k_1}n))=\rem_2(\mathbf{i}(2^{k_2}n))$ for all $n\geq0$.
However, taking~$n=2^{k_1+1}-1$, we have $\rem_2(\mathbf{i}(2^{k_{1}}n))=\rem_2(\mathbf{i}(2^{k_{1}}(2^{k_{1}+1}-1)))=1$, but~$\rem_2(\mathbf{i}(2^{k_{2}}n))=\rem_2(\mathbf{i}(2^{k_{2}}(2^{k_{1}+1}-1)))=0$, which is a contradiction.

In fact, for every $n\geq0$, $$\mathbf{i}(n)=\left\{
       \begin{array}{ll}
         2s_2(n)-|(n)_2|, & \hbox{if $2s_2(n)>|(n)_2|$,} \\
         0, & \hbox{otherwise.}
       \end{array}
     \right.$$
That is, $\mathbf{i}(n)$ equals $0$ if the number of $0$'s is more than
$1$'s in the $(n)_2$, otherwise it equals the difference between the number of $1$'s and $0$'s in the $(n)_2$.

Similarly, the index sequence of $\sigma_2:q_0\rightarrow q_0q_1$, $q_i\rightarrow q_{i+1}q_{i-1}~(i\geq1)$ is not $2$-regular.

\medskip

Let $b_j\in\mathbb{N}$ and $\sigma:~q_i\rightarrow q_{i}q_{i+b_1}\cdots q_{i+b_{k-1}}$ be a $1$-periodic~$k$-uniform morphism. By Theorem \ref{index regular}, its index sequence is $k$-regular. For every~$i\geq0$,we define
\begin{equation}\label{1-periodic}
  \sigma:~q_i\rightarrow q_{a_0i}q_{a_1i+b_1}\cdots q_{a_{k-1}i+b_{k-1}}
\end{equation}
where~$a_j\in\mathbb{N}$. Then, we have the following proposition.
\begin{proposition}\label{lemma5555}
The index sequence~$\{\mathbf{i}(n)\}_{n\geq0}$~of $\sigma$ defined by (\ref{1-periodic}) is~$k$-regular.
\end{proposition}
\begin{proof}
By the definition of $\sigma$, we have $\sigma(q_{\mathbf{i}(n)})=q_{a_0\mathbf{i}(n)}q_{a_1\mathbf{i}(n)+b_1}\cdots q_{a_{k-1}\mathbf{i}(n)+b_{k-1}}$. Since $\sigma$ is $k$-uniform, we also have $\sigma(q_{\mathbf{i}(n)})=q_{\mathbf{i}(kn)}q_{\mathbf{i}(kn+1)}\cdots q_{\mathbf{i}(kn+k-1)}$. Hence, for every $0\leq j\leq k-1$,
$\mathbf{i}(kn+j)=a_j\mathbf{i}(n)+b_j$ for all~$n\geq0$, where
$b_0=0$, which completes our proof.
\end{proof}

For every $0\leq n\leq m-1$ and $i\geq0$, we define a morphism to be $\sigma$ by $$\sigma(q_{mi+n})=q_{a_{n,0}\cdot mi+t_{n,0}}q_{a_{n,1}\cdot mi+t_{n,1}}\cdots q_{a_{n,k-1}\cdot mi+t_{n,k-1}}$$
where $A=(a_{r,s})_{0\leq r\leq m-1,0\leq s\leq k-1}$ and $T=(t_{r,s})_{0\leq r\leq m-1,0\leq s\leq k-1}$\
$\in\mathbb{N}^{m\times k}$. Denote this morphism by $\sigma_{A,T}$ briefly.  If all elements of $A$ are 1, then the morphism $\sigma_{A,T}$ is a $m$-periodic~$k$-uniform morphism. Hence, by Theorem \ref{index regular}, its index sequence is $k$-regular. However, for a general matrix $A$, we do not know that whether the index sequence of this morphism $\sigma_{A,T}$
is $k$-regular or not.

\medskip

Until now, all the morphisms we considered are linear. The following example gives a non-linear morphism and shows that it is not $2$-regular.
\begin{example}\label{ex8}
Let $\alpha>1$ be a real number. The index sequence of $\sigma_2:q_i\rightarrow q_{i^{\alpha}}q_2~(i\geq0)$ is not $2$-regular.
\end{example}
In fact, let~$\{\mathbf{i}(n)\}_{n\geq0}$ be the index sequence of this morphism~$\sigma_2$. Then~$\mathbf{i}(0)=0,~\mathbf{i}(1)=2$, $\mathbf{i}(2n)=(\mathbf{i}(n))^{\alpha}$~and $\mathbf{i}(2n+1)=2$ for all $n\geq0$. Hence, $\mathbf{i}(2^k)=2^{\alpha^k}$ and
$$\frac{\log_2(\mathbf{i}(2^k))}{\log_2(2^k)}=\frac{\alpha^k}{k}\rightarrow\infty~~~~(k\rightarrow\infty).$$
Thus, by Theorem 2.10 in \cite{Jp}, the sequence~$\{\mathbf{i}(n)\}_{n\geq0}$~is not~$2$-regular.

\section{Regularity of the sequences generated by linear recurrence codings}\label{section}
Given an integer sequence $L=\{L_n\}_{n\geq0}$, we define a coding $\tau_{L}$ to be $\tau_{L}(q_i)=L_i$ for all $i\geq0$.
In particular, if $L=\{n\}_{n\geq0}$, then $\tau_{L}$ is the coding under consideration in Section $3$.
If there exist integers $p\geq0, C_0,C_1,C_2,\cdots,C_p$~satisfying that
$$L(n)=\sum_{i=1}^pC_iL(n-i)+C_0~~~(n\geq p),$$ then we say the sequence~$\{L_n\}_{n\geq0}$~satisfies a \emph{linear recurrence}.
In this section, we will consider the coding $\tau_{L}$ satisfying a linear recurrence and its corresponding sequences.

Let~$\{u(n)\}_{n\geq0}$ be an integer sequence. If there exist integers~$l\geq0$~and~$0\leq j\leq k^l-1$ such that for every
$l^{\prime}>l$ and $0\leq j^{\prime}\leq k^{l^{\prime}}-1$,
\begin{equation*}\label{indexsequence}
 u(k^{l^{\prime}}n+j^{\prime})=u(k^{l}n+j)+c_{l^{\prime},j^{\prime}}
\end{equation*}
where~$c_{l^{\prime},j^{\prime}}$~is a constant, depending on~$l^{\prime},j^{\prime}$, then the sequence $\{u(n)\}_{n\geq0}$ is called \emph{$(l,j)$-order recursive}. If there exist integers~$l\geq0$~and~$0\leq j\leq k^l-1$ satisfying that $\{u(k^{l}n+j):n\geq0\}\supset\mathbb{N}$, then the sequence $\{u(n)\}_{n\geq0}$ is called \emph{$(l,j)$-order complete}. For example, the sequence $\{S_2(n)\}_{n\geq0}$ is both $(l,j)$-order recursive and $(l,j)$-order complete with $l=j=0,c_{l^{\prime},j^{\prime}}=s_2(j^{\prime})$, since $S_2(2n)=S_2(n)$ and $S_2(2n+1)=S_2(n)+1$.

Note that, for every $l\geq0$~and~$0\leq j\leq k^l-1$, a $(l,j)$-order recursive integer sequence is $k$-regular. In particular, we have the following theorem.
\begin{theorem}\label{thm4.2}
Let $k\geq2$ 
be an integer and $\{u(n)\}_{n\geq0}$ be a non-negative integer sequence. If $\{u(n)\}_{n\geq0}$ is both $(0,0)$-order recursive and $(0,0)$-order complete, then $\{L_n\}_{n\geq0}$~satisfies a linear recurrence if and only if $\{L_{u(n)}\}_{n\geq0}$ is $k$-regular.
\end{theorem}
\begin{remark}
In fact, if there exist integers $l\geq0, 0\leq j\leq k^{l}-1$ satisfying that $\{u(n)\}_{n\geq0}$ is $(l,j)$-order recursive and $\{L_n\}_{n\geq0}$~satisfies a linear recurrence, then $\{L_{u(n)}\}_{n\geq0}$ is also $k$-regular.
\end{remark}
\begin{remark}
Theorem \ref{thm4.2} tells us that the images of some regular sequences under linear recurrence codings are also regular.
\end{remark}

To prove Theorem \ref{thm4.2}, we need the following two lemmas.
\begin{lemma}\label{recursive}
If $\{u(n)\}_{n\geq0}$ is $(0,0)$-order recursive, then for every
$l>0$ and $0\leq j\leq k^{l}-1$, $$c_{l,j}=\sum_{s=0}^{l-1}c_{1,j_s}$$ where $(j)^{l}_{k}=j_{l-1}\cdots j_{1}j_{0}.$
\end{lemma}
\begin{proof}
This is an immediate consequence of the fact that $u(kn+j)=u(n)+c_{1,j}$ for $0\leq j\leq k-1.$
\end{proof}
\begin{lemma}\label{complete}
If $\{u(n)\}_{n\geq0}$ is $(0,0)$-order recursive and $(0,0)$-order complete, then $\max\{c_{1,j}:0\leq j\leq k-1\}>0$. Moreover, if $\{u(n)\}_{n\geq0}$ is non-negative, then, $c_{l,j}\geq0$ for every $l>0$ and $0\leq j\leq k^{l}-1.$
\end{lemma}
\begin{proof}
If $\max\{c_{1,j}:0\leq j\leq k-1\}\leq 0$, then $c_{1,j}\leq0$ for all $0\leq j\leq k-1$. By Lemma \ref{recursive}, we have $c_{l,j}\leq0$ for all $l>0$ and $0\leq j\leq k^{l}-1$. Hence, if $(n)_{k}=j_{l-1}\cdots j_{1}j_{0}$, then $u(n)=u(0)+\sum_{s=0}^{l-1}c_{1,j_s}\leq u(0)$, which contradicts the $(0,0)$-order completeness of $\{u(n)\}_{n\geq0}$.

If there exists an integer $0\leq r\leq k-1$ such that $c_{1,r}<0$, choosing a nature number $l$ such that $u(0)+lc_{1,r}<0$, then $u(n)=u(0)+\sum_{s=0}^{l-1}c_{1,r}< 0$ where $(n)_{k}=r^l$. Since $\{u(n)\}_{n\geq0}$ is a non-negative sequence, it implies that $c_{1,j}\geq0$ for every $0\leq j\leq k-1$. By Lemma \ref{recursive}, $c_{l,j}\geq0$, for every $l>0$ and $0\leq j\leq k^{l}-1.$
\end{proof}
Now, we are going to prove Theorem \ref{thm4.2}.
\begin{proof}[Proof of Theorem \ref{thm4.2} ]
Assume that the sequence $\{L_n\}_{n\geq0}$~satisfies a linear recurrence, i.e.,  there exists an integer~$p\geq1$~such that $L_n=\sum_{i=1}^pC_iL_{n-i}+C_0~(n\geq p)$, where~$C_i~(0\leq i\leq p)$~are constants. Then, for every integer $c\geq0$, $L_{n+c}=\sum_{i=1}^pC^{\prime}_i(c)L_{n-i}+C^{\prime}_0(c)~(n\geq p)$, where~$C^{\prime}_i(c)~(0\leq i\leq p)$~are constants.
Hence, by Lemma \ref{complete} and the $(0,0)$-order recursive relation, for~$l\geq 0$ and $0\leq j\leq k^{l}-1$,
we have
$$L_{u(k^{l}n+j)}=L_{u(n)+c_{l,j}}
   = \sum_{i=1}^pC_iL_{u(n)+c_{l,j}-i}+C_0
   = \sum_{i=1}^pC_i^{\prime}(l,j)L_{u(n)-i}+C^{\prime}_0(l,j),
$$
where for~$C_i^{\prime}(l,j)~(0\leq i\leq p)$~are all constants, depending on $l$ and $j$. Thus,  each sequence $\{L_{u(k^{l}n+j)}\}_{n\geq0}$
is a combination of the sequences $\{L_{u(n)-1}\}_{n\geq0}$,
$\{L_{u(n)-2}\}_{n\geq0},\cdots,\{L_{u(n)-p}\}_{n\geq0}$~and
the constant sequence, which implies that the sequence~$\{L_{u(n)}\}_{n\geq0}$~is~$k$-regular.

Conversely, if $\{L_{u(n)}\}_{n\geq0}$~is a $k$-regular sequence, then the module generated by its~$k$-kernel is  generated by $\{L_{u(k^{l_s}n+j_s)}\}_{n\geq0}~(1\leq s\leq M)$. Assume $l>\max\{l_s:1\leq s\leq M\}$ and $c_{1,r}=\max\{c_{1,j}: 0\leq j\leq k-1\}$ for some $0\leq r\leq k-1$. By Lemma \ref{complete}, $c_{1,r}>0$. Taking $j=[r^{l}]_k=r\cdot\frac{k^{l}-1}{k-1}$. Then, by Lemma \ref{recursive}, we have $u(k^{l}n+j)=u(n)+l c_{1,r}.$ Hence, we have
$$L_{u(n)+lc_{1,r}}=L_{u(k^{l}n+j)}
   = \sum_{s=1}^{M}C_s(l,j)L_{u(k^{l_{s}}n+j_{s})}
   = \sum_{s=1}^{M}C_s(l,j)L_{u(n)+c_{l_s,j_s}}.
$$
By the choice of $l$, note that  for all $1\leq s\leq M$ and $lc_{1,r}> c_{l_s,j_s}$. Let $m=u(n)+lc_{1,r}$. Then, $$L_{m}=\sum_{s=1}^{M}C_s(l,j)L_{m-(lc_{1,r}-c_{l_s,j_s})}.$$
Since $u(n)$ is $(0,0)$-order complete, $m$ ranges all natural numbers except finitely many terms. Hence, $\{L_n\}_{n\geq0}$~satisfies a linear recurrence.
\end{proof}

\begin{example}\label{ex9}
Let $\sigma:q_i\rightarrow q_iq_{i+1}$ and $\{F_n\}_{n\geq0}$~be the Fibonacci numbers defined by~$F_0=1,F_1=1$ and $F_n=F_{n-1}+F_{n-2}$ for every $n\geq2$. Note that the index sequence $\{\mathbf{i}(n)\}_{n\geq0}$ of $\sigma$ is $(0,0)$-order recursive and $(0,0)$-order complete, $\{F_n\}_{n\geq0}$ satisfies a linear recurrence. Hence, the sequence $\tau_{F}(\sigma^{\infty}(q_0))=\{F_{\mathbf{i}(n)}\}_{n\geq0}$ is $2$-regular.
\end{example}

It is clearly that the index sequence $\{\mathbf{i}(n)\}_{n\geq0}$ is a $2$-regular sequence which can be generated by $\mathbf{i}(2n)=\mathbf{i}(n)$ and $\mathbf{i}(2n+1)=\mathbf{i}(n)+1$ with $\mathbf{i}(0)=0$. Under the coding $\rho:i\rightarrow F_i$, we obtain a new $2$-regular sequence $\rho(\mathbf{i}(n))=\{F_{\mathbf{i}(n)}\}_{n\geq0}=\tau_{F}(\sigma^{\infty}(q_0))$ which can be generated by the formulas:
$$ F_{\mathbf{i}(0)}=F_{\mathbf{i}(1)}=1,F_{\mathbf{i}(2n)}=F_{\mathbf{i}(n)},F_{\mathbf{i}(4n+1)}=F_{\mathbf{i}(2n+1)},F_{\mathbf{i}(4n+3)}=F_{\mathbf{i}(2n+1)}+F_{\mathbf{i}(n)}.$$

\bigskip

Let $\sigma:~q_i\rightarrow q_{i+t_0}q_{i+t_1}\cdots q_{i+t_{k-1}}$ be a~$1$-periodic~$k$-uniform morphism with $t_0=0,t_i\in\mathbb{N}~(0\leq i\leq k-1)$.  Assume that $\{\mathbf{i}(n)\}_{n\geq0}$~is the index sequence of $\sigma$, then $\mathbf{i}(kn+j)=\mathbf{i}(n)+t_{j}$ for all $0\leq j\leq k-1$. Note that $\{\mathbf{i}(n)\}_{n\geq0}$ is $(0,0)$-order recursive with $c_{1,j}=t_j$ and $\max\{c_{1,j}:0\leq j\leq k-1\}>0$. Hence, by Theorem \ref{thm4.2}, we have the following corollary.
\begin{cor}\label{cor1}
If there exist two integers $1\leq i,j\leq k-1$ such that $(t_i,t_j)=1$, then the sequence $\{L_n\}_{n\geq0}$~satisfies a linear recurrence  if and only if the sequence  $\tau_{L}(\sigma^{\infty}(q_0))$ is $k$-regular.
\end{cor}
\begin{proof}
Since $\tau_{L}(\sigma^{\infty}(q_0))=\{\tau_{L}(q_{\mathbf{i}(n)})\}_{n\geq0}=\{L_{\mathbf{i}(n)}\}_{n\geq0}$, by Theorem \ref{thm4.2}, it suffices to show that $\{\mathbf{i}(n)\}_{n\geq0}$ is $(0,0)$-order complete. Assume that $0<t_i=p<q=t_j$ satisfies $(p,q)=1$ for some $0\leq i,j\leq k-1.$ Now, we claim that for every large integer $m\geq0$, there exist integers $c_1,c_2\geq0$ such that $m=c_1 p+c_2q$.

Assume $m=pn+i$ with $0\leq i\leq p-1$, then $m=p(n-j)+(pj+i)$ for every $j\geq 0$. Since $(p,q)=1$, $\{pj+i:0\leq j\leq q-1\}$ forms a complete system of incongruent residues $(\bmod~q$). Hence, there exist integers $0\leq j_1\leq q-1$ and $j_2\geq0$ such that $pj_1+i=qj_2.$ Thus, $m=(n-j_1) p+j_2q$ which completes the claim.

For every $0\leq i,j\leq k-1$, let $n=[i^{c_1}j^{c_2}]_k$. Then by Lemma \ref{recursive}, $\mathbf{i}(n)=\mathbf{i}(0)+c_1t_i+c_2t_j=\mathbf{i}(0)+c_1 p+c_2q.$ Hence, $\{\mathbf{i}(n)\}_{n\geq0}$~is $(0,0)$-order complete except finitely many terms.
\end{proof}
\begin{remark}
The condition of Corollary \ref{cor1} can be weakened by the condition that the greatest common factor of all nonzero integers $t_i~(0\leq i\leq k-1)$ is $1$.
\end{remark}

Let $\sigma:~q_i\rightarrow q_0q_{i+t_1}q_{i+t_2}\cdots q_{i+t_{k-1}}~(i\geq0)$ be a~$k$-uniform morphism,
where~$t_i\geq0~(1\leq i\leq k-1)$~are integers. Assume that~$\{\mathbf{i}(n)\}_{n\geq0}$~is the index sequence of this morphism, then $\mathbf{i}(kn)=0$~and~$\mathbf{i}(kn+j)=\mathbf{i}(n)+t_j~(1\leq j\leq k-1)$ for every $n\geq0$. Note that the sequence $\{\mathbf{i}(n)\}_{n\geq0}$ is not $(0,0)$-order recursive.

The following proposition shows that if $\{u(n)\}_{n\geq0}$ is not $(0,0)$-order recursive, then, Theorem \ref{thm4.2} also holds.

\begin{proposition}\label{thm}
If there exist two integers $1\leq i,j\leq k-1$ such that $(t_i,t_j)=1$, then the sequence $\{L_n\}_{n\geq0}$~satisfies a linear recurrence  if and only if the sequence  $\tau_{L}(\sigma^{\infty}(q_0))=\{L_{\mathbf{i}(n)}\}_{n\geq0}$ is $k$-regular.
\end{proposition}
\begin{proof}
For every $l\geq0,0\leq j\leq k^{l}-1$, let $(j)^{l}_{k}=j_{l-1}j_{l-2}\cdots j_0.$ If $j_i\neq0$ for all $0\leq i\leq l-1$, then $\mathbf{i}(k^{l}n+j)=\mathbf{i}(n)+\sum_{i=0}^{l-1}t_i$. Otherwise, assume that $s=\min\{0\leq i\leq l-1:j_i=0\}$, then $\mathbf{i}(k^{l}n+j)=\sum_{i=0}^{s}t_i$.

If the sequence~$\{L_n\}_{n\geq0}$
satisfies a linear recurrence, i.e., there exists an integer~$p\geq1$~satisfying that
$L_n=\sum_{i=1}^pC_iL_{n-i}+C_0~(n\geq p)$, where~$C_i~(0\leq i\leq p)$~are constants. Then, for every integer $c\geq0$, $L_{n+c}$ is a combination of $L_{n-i}$ with $1\leq i\leq p$. Hence, for every $l\geq0$ and $0\leq j\leq k^{l}-1$, the sequence $\{L_{\mathbf{i}(k^{l}n+j)}\}_{n\geq0}$ is a combination of $\{L_{\mathbf{i}(n)-i}\}_{n\geq0}~(1\leq i\leq p)$ and the constant sequence, which implies the sequence~$\{L_{\mathbf{i}(n)}\}_{n\geq0}$~is~$k$-regular.

The other direction is the same as the proof of Theorem \ref{thm4.2} and Corollary \ref{cor1}, so we omit it here.
\end{proof}

Let $\sigma$ be the morphism defined by formula (\ref{1-periodic}), i.e., $\sigma:~q_i\rightarrow q_{a_0i}q_{a_1i+b_1}\cdots q_{a_{k-1}i+b_{k-1}}$ with $a_i\in\mathbb{N}.$ Corollary \ref{cor1} and Proposition \ref{thm} have studied the cases $a_i\in\{0,1\}$ for all $0\leq i\leq k-1$. If there exists $0\leq i\leq k-1$ such that $a_i\geq2$, then the following example shows that Theorem \ref{thm4.2} does not hold.
\begin{example}\label{ex10}
For every $i\geq0$, let $\sigma:q_i\rightarrow q_{2i}q_{i+1},~\tau:q_i\rightarrow F_i$, where
$\{F_n\}_{n\geq0}$ is the Fibonacci numbers. Then the sequence $\tau(\sigma^{\infty}(q_0))$
is not $2$-regular.
\end{example}
Let $\{\mathbf{i}(n)\}_{n\geq0}$ be the index sequence of $\sigma$ in Example \ref{ex3}. Then $\mathbf{i}(2n)=2\mathbf{i}(n)$ and $\mathbf{i}(2n+1)=\mathbf{i}(n)+1$ with
$\mathbf{i}(0)=0$. Note that $F_{2n}=F^2_{n}+2F_{n}F_{n-1}$ and assume that $d_n=F_{\mathbf{i}(n)}$, then,
$$d_{2n}=F_{\mathbf{i}(2n)}=F_{2\mathbf{i}(n)}=F^2_{\mathbf{i}(n)}+2F_{\mathbf{i}(n)}F_{\mathbf{i}(n)-1}>F^2_{\mathbf{i}(n)}=d^2_n.$$
Since $d_2=2$ and
$$\frac{\log_2(d(2\cdot2^k))}{\log_2(2\cdot2^k)}>\frac{2^k \log_2 d_2}{k+1}=\frac{2^k}{k+1}\rightarrow\infty~~~~(k\rightarrow\infty).$$
By Theorem 2.10 in \cite{Jp}, the sequence~$\tau(\sigma^{\infty}(q_0))=\{F_{\mathbf{i}(n)}\}_{n\geq0}=\{d_n\}_{n\geq0}$~is not~$2$-regular .

\bigskip

In Theorem \ref{thm4.2}, if the sequence $\{L_n\}_{n\geq0}$ takes only finitely many values, then, we have the following proposition.
\begin{proposition}\label{main11}
Let $k\geq2$ 
be an integer and $\{u(n)\}_{n\geq0}$ be a non-negative integer sequence. If the sequence $\{u(n)\}_{n\geq0}$ is both $(0,0)$-order recursive and $(0,0)$-order complete, then the sequence $\{L_{u(n)}\}_{n\geq0}$ is $k$-automatic if and only if
the sequence $\{L_n\}_{n\geq0}$ is ultimately periodic.
\end{proposition}
\begin{proof}
Note that if a sequence takes only finitely many values, then Everest et al. in \cite{GAI} told us that it satisfies a linear recurrence if and only if it is ultimately periodic, and Allouche and Shallit in \cite{Jp} showed that it is regular if and only if it is $k$-automatic. Then, by Theorem \ref{thm4.2}, we complete this proof.
\end{proof}
\begin{remark}
Proposition \ref{main11} shows that we can obtain an automatic sequence from a ultimately periodic sequence by a $``u(n)"$-index picking, where $\{u(n)\}_{n\geq0}$ is both $(0,0)$-order recursive and $(0,0)$-order complete.
\end{remark}
It is known that the Fibonacci sequence $\{f_n\}_{n\geq0}=\psi^{\infty}(0)$ is not ultimately periodic, where $\psi:0\rightarrow01,1\rightarrow0$. Hence, if $\sigma:~q_i\rightarrow q_{i}q_{i+t_1}\cdots q_{i+t_{k-1}}$ and $(t_i,t_j)=1$ for some $1\leq i,j\leq k-1$, then by Proposition \ref{main11}, $\tau_{f}(\sigma^{\infty}(q_0))$ is not $k$-automatic.

\section{Some generalizations}
Both the alphabet $\Sigma$ and the state set $Q$ in Section 3 and Section 4 are denoted by $\{q_n:n\in\mathbb{N}\}$ and the index sequences take values in $\mathbb{N}$.
A possible extension of the present approach is to focus on the infinite set $\{q_n:n\in\mathbb{Z}\}$ and the index sequences on $\mathbb{Z}$. On the infinite set $\{q_n:n\in\mathbb{Z}\}$, we can define morphisms and DCAs (DCAOs) similarly.
Moreover, we generalize $m$-periodic~$k$-uniform morphisms (or the $m$-periodic~$k$-DCAOs) by $T\in \mathbb{Z}^{m\times k}.$

Let $\sigma$ be a $m$-periodic $k$-uniform morphism and $\sigma^{\infty}(q_0)$ be a fixed point of $\sigma$. If $\tau(q_i)=i$, then the sequence $\{\mathbf{i}(n)\}_{n\geq0}:=\tau(\sigma^{\infty}(q_0))$ is called to be the \emph{index sequence} of $\sigma$. The following results are similar as Theorem \ref{index regular} and Proposition \ref{main2}.

\begin{thmbis}{index regular}
If $\sigma$ is a $m$-periodic $k$-uniform morphism, then its index sequence~$\{\mathbf{i}(n)\}_{n\geq0}$~is~$k$-regular.
\end{thmbis}

\begin{propbis}{main2}
Taking $\tau:q_i\rightarrow f(i)$~for every $i\geq0$, where
$f(i)$ is a polynomial with integer coefficients. If $\sigma$ is a $m$-periodic $k$-uniform morphism, then the sequence~$\tau(\sigma^{\infty}(q_0))$ is $k$-regular.
\end{propbis}

We give the following $2$-periodic $2$-uniform morphism on $\{q_n:n\in\mathbb{Z}\}$.
\begin{example}
Let $\sigma$ be a morphism defined by $q_{2i}\mapsto q_{2i}q_{2i+1}, q_{2i+1}\mapsto q_{2i-1}q_{2i+2} (i\in\mathbb{Z})$. Then, $\sigma^{\infty}(q_0)$ is the unique fixed point of $\sigma$. Let $\tau(q_i)=i$. Then, the index sequence $\{\mathbf{i}(n)\}_{n\geq0}=\tau(\sigma^{\infty}(q_0))=01(-1)2(-3)023(-5)(-2)012314\cdots$. It also can be generated by a DCAO in Figure \ref{integer}.
\begin{figure}[H]
\centering
\begin{tikzpicture}[scale=0.8, every node/.style={scale=0.8}, state/.style={scale=0.8, circle solidus,draw,
inner sep=1pt,minimum size=12mm},>=stealth,->,auto,black]
\node[state,initial,initial where=above] (a) {$q_0$\nodepart{lower} $0$};
\node[state] (b) [right=15mm of a] {$q_1$ \nodepart{lower} $1$};
\node[state] (c) [right=15mm of b] {$q_2$ \nodepart{lower} $2$};
\node(d) [right=15mm of c] {$\cdots$};
\node[state] (e) [left=15mm of a] {$q_{-1}$ \nodepart{lower} $-1$};
\node[state] (f) [left=15mm of e] {$q_{-2}$ \nodepart{lower} $-2$};
\node(g) [left=15mm of f] {$\cdots$};

\draw [->] (b.20) to [bend left]  node [above] {$1$} (c.160);
\draw [->] (c.20) to [bend left]  node [above] {$1$} (d.160);
\draw [->] (b.260) to [bend left=100] node [below] {$0$} (e.-90);
\draw [->] (d.280) to [bend left=100] node [below] {$0$} (b.-80);
\draw [->] (a.20) to [bend left]  node [above] {$1$} (b.160);
\draw [->] (a) to [in=-60, out=-120,loop,distance=10mm] node [below] {$0$} (a);
\draw [->] (c) to [in=-60, out=-120,loop,distance=10mm] node [below] {$0$} (c);
\draw [->] (f) to [in=-60, out=-120,loop,distance=10mm] node [below] {$0$} (f);
\draw [->] (e.20) to [bend left]  node [above] {$1$} (a.160);
\draw [->] (e.250) to [bend left=100] node [below] {$0$} (g.-90);
\draw [->] (f.20) to [bend left]  node [above] {$1$} (e.160);
\draw [->] (g.20) to [bend left]  node [above] {$1$} (f.160);
\end{tikzpicture}
\caption{A DCAO generating the sequence $\{b_n\}_{n\geq0}$.}
\label{integer}
\end{figure}
\end{example}
We end this section by two regular sequences which can be generated by countable state automata and also morphisms on a countable alphabet.
\begin{example}
Let $a(n)=\lfloor\log_b(\alpha n+\beta)\rfloor$. Then, the integer sequence $\{a(n)\}_{n\geq0}$ has been studied by Allouche, Shallit, Bell, Moshe, Rowland and Zhang et al. in \cite{jp,YM,Row,Zhang} respectively. Especially, if $b=2,\alpha=1,\beta=0$, we obtain a $2$-regular sequence $\{\lfloor\log_2 n\rfloor\}_{n\geq1}$. Let $b_0=1,b_n=\lfloor\log_2 n\rfloor~(n\geq1)$. Then, it can be generated by formulas: $b_0=1,b_1=0,b_{2n}=b_{2n+1}=b_n+1~(n\geq1)$. It also can be generated by a DCAO in Figure \ref{figure1111}.
\end{example}
\begin{figure}[H]
\centering
\begin{tikzpicture}[scale=0.8, every node/.style={scale=0.8}, state/.style={scale=0.8, circle solidus,draw,
inner sep=1pt,minimum size=12mm},>=stealth,node distance=2.5cm,->,auto,black]
    \node[state,initial]  (1)                      {$q_0$\nodepart{lower}$0$};
    \node[state] (2) [right of=1] {$q_1$\nodepart{lower}$0$};
    \node[state](3) [right of=2] {$q_2$\nodepart{lower}$1$};
    \node(4) [right of=3] {$\cdots$};
    \node[state] (n) [right of=4] {$q_{n+1}$\nodepart{lower}$n$};
    \node(n+1) [right of=n] {$\cdots$};

    \path[->]
    		(1) edge [bend left]  node {$1$} (2)
                      edge [loop above] node {$0$} ()
    		(2) edge [bend left]   node  {$0,1$} (3)	
            (3) edge [bend left]   node  {$0,1$} (4)
            (4) edge [bend left]   node  {$0,1$} (n)
            (n) edge [bend left]   node  {$0,1$} (n+1);
    \end{tikzpicture}
\caption{A DCAO generating the sequence $\{b_n\}_{n\geq0}$.}
\label{figure1111}
\end{figure}

\begin{example}
Let $f(n)$ be a polynomial with integer coefficients. Then, $\{f(n)\}_{n\geq0}$ is a $k$-regular sequence for every $k
\geq2$ in \cite{Jp}. Especially, the sequence~$\{b_n\}_{n\geq0}=\{n\}_{n\geq0}$ is $2$-regular, since it can be generated by formulas: $b_{2n}=2b_n$ and $b_{2n+1}=2b_n+1$. It also can be generated by $2$-uniform morphism~$\sigma:i\rightarrow(2i)~(2i+1)~(i\geq0)$, i.e., $\{b_n\}_{n\geq0}=\sigma^{\infty}(0)$.
\end{example}

\medskip

\noindent\textbf{Acknowledgements. }

The authors wish to thank  Professor Zhi-Ying Wen for inviting them to visit
the Morningside Center of Mathematics, Chinese Academy of Sciences.
They also wish to thank Professor Li-Feng Xi for his helpful suggestions.

\end{document}